\documentclass[11pt]{article}

\usepackage[T1]{fontenc}
\usepackage[utf8]{inputenc}
\usepackage[margin=1in]{geometry}
\usepackage{amsmath,amssymb,amsthm,mathtools}
\usepackage{enumitem}
\usepackage{booktabs,array,graphicx}
\usepackage[hidelinks]{hyperref}

\theoremstyle{plain}
\newtheorem{theorem}{Theorem}[section]
\newtheorem{lemma}[theorem]{Lemma}
\newtheorem{proposition}[theorem]{Proposition}
\newtheorem{corollary}[theorem]{Corollary}
\theoremstyle{definition}
\newtheorem{definition}[theorem]{Definition}
\newtheorem{remark}[theorem]{Remark}

\newcommand{\V}{V}
\newcommand{\Reach}{\operatorname{\mathsf{Reach}}}
\newcommand{\cl}[1]{\lceil #1 \rceil}
\newcommand{\res}{\mathsf{R}}
\newcommand{\miss}{\operatorname{\mathsf{Miss}}}
\newcommand{\Serve}{\operatorname{\mathsf{Serve}}}
\newcommand{\restr}{\!\upharpoonright\!}
\newcommand{\NP}{\textsf{NP}}
\newcommand{\ALG}{\textup{\textsc{Alg}}}
\newcommand{\OPT}{\textup{\textsc{Opt}}}

\title{\bfseries Service-Cut Certificates for Aligned Eviction\\
in Tiered Cache Networks}

\author{
Faruk Alpay\thanks{Corresponding author: \texttt{alpay@lightcap.ai}.} \qquad Levent Sarıoğlu\\[5pt]
\normalsize Department of Computer Engineering, Bah\c{c}e\c{s}ehir University\\
\normalsize Istanbul, Turkey\\
\normalsize \texttt{faruk.alpay@bahcesehir.edu.tr, levent.sarioglu@bahcesehir.edu.tr}
}

\date{}

\begin{document}
\maketitle

\begin{abstract}
In a tiered cache, eviction is a graph decision: removing one aligned storage block can disconnect downstream demand that never addressed that block directly, so request recency alone cannot price the action.  This paper studies aligned eviction as a vertex-separation problem and gives a selection rule whose decisions carry independently checkable service-cut evidence.  For every candidate block, it computes the exact weighted downstream demand cut, rejects actions that disconnect protected demand, and selects the minimum-impact admissible eviction.  Reclamation is characterized as vertex separation: minimum-location reclamation reduces to node-capacitated flow, while minimum aligned block actions are \NP-complete.  In two-hop cache networks, one streaming pass evaluates every candidate impact; a matching adversarial construction proves that a history-only victim selector has unbounded one-step damage.  The packet-scale implementation combines a seed-indexed exact-cardinality residency structure with collision-aware, 32-bank impact counters.  Replay compression makes the result auditable: counter intervals reproduce the stream, exact monoid summaries retain every reported additive statistic, and a counting lower bound quantifies the state required by any exact all-candidate summary.  A 144-scenario evaluation processes 582.90 trillion packets (404.86~PiB of simulated payload), validates the coordinate expectations, and exposes a zero-impact extreme-value transition near $N\zeta=\log m$.  Complete impact vectors, decoded audit samples, telemetry, and logs remain within the ancillary-file budget.  Finally, invalidation is monotone replicated state: fair asynchronous delivery converges without coordination, with a diameter bound under synchronous full-edge rounds.  The architecture therefore binds capacity reclamation, path continuity, and distributed invalidation to one certifying interface.
\end{abstract}

\section{Introduction}
\label{sec:intro}

A tiered cache network couples storage placement to service reachability.  An origin retains the authoritative copy, intermediate locations hold resident replicas, and directed transfer links determine which replicas can serve downstream demand.  The near tier is capacity-limited, but its eviction unit is a fixed-size aligned block rather than an isolated object.  Reclaiming one block may therefore remove an intermediate service cut and disconnect demand far from the physical storage action.  An eviction layer that observes only request recency cannot state, let alone certify, this graph-level consequence.

Request recency alone does not measure this damage.  A block that has not been requested recently may be the only resident parent of a large demand set.  Conversely, a frequently accessed block may be redundant because the same demand has several resident parents.  The eviction problem therefore has two parts: identify the graph cut induced by a candidate block, and expose enough evidence that another implementation can check the decision.

The design objective is a certifying interface: before mutating residency, expose the exact downstream service loss, prove that protected demand remains connected, and make the resulting invalidation converge across replicas without global coordination.  The construction draws on caching networks~\cite{ioannidis,baev}, vertex connectivity~\cite{menger,schrijver}, online paging~\cite{sleatortarjan,borodin}, certifying algorithms~\cite{certifying}, and monotone replicated state~\cite{calm,crdt,delta}.

\paragraph{Contributions.}
\begin{enumerate}[leftmargin=2em,itemsep=2pt]
\item A shared reachability certificate represents the exact pre- and post-eviction service sets.  It avoids one certificate per disconnected endpoint, has linear size, and verifies in linear time for one candidate (Lemma~\ref{lem:impact}).
\item Reclamation is characterized as vertex separation.  Minimum-location reclamation is a node-capacitated flow problem, minimum aligned block actions are \NP-complete, and alignment yields an exact reclaimable-capacity threshold (Theorems~\ref{thm:menger}, \ref{thm:hardness}, and \ref{thm:granularity}).
\item A one-pass accumulator evaluates all candidate blocks simultaneously in two-hop cache networks (Theorem~\ref{thm:onepass}).  A topology-oblivious selector can have an arbitrarily large one-step damage ratio, whereas minimum impact is optimal for the current state (Theorems~\ref{thm:separation} and \ref{thm:perstep}).
\item A constant-state affine index enforces exact resident cardinality in both demand strata.  Warp grouping preserves exact impact vectors but has only collision-dependent savings, quantified by Corollary~\ref{cor:warp-sparse}; sharded counter banks remove address conflicts and yield the measured crossover gain (Theorem~\ref{thm:exact-residency} and Propositions~\ref{prop:warp}--\ref{prop:banks}).  Replay compression then makes the packet experiment compatible with a small research artifact: a seed and counter interval reproduce the stream, and an exact monoid summary preserves every reported additive metric (Theorems~\ref{thm:replay} and \ref{thm:summary-lb}).  The artifact retains full candidate vectors, decoded samples, telemetry, logs, and independent validation.
\end{enumerate}

The distributed invalidation result is stated with the scheduling assumptions made explicit: fair asynchronous delivery gives eventual convergence, while a diameter bound requires synchronous full-edge rounds.  Recoverable and irreversible eviction are handled by separate state transitions rather than conflated with cache replacement.

\section{Transfer networks and aligned eviction}
\label{sec:prelim}

\begin{definition}[Transfer network]
\label{def:net}
A \emph{transfer network} is a tuple $\mathfrak N=(\V,E,\Pi,\rho)$ where $\V$ is a finite nonempty set of \emph{locations}, $E\subseteq \V\times \V$ is a \emph{transfer relation} ($u\to v$ means a resident copy at $u$ may populate or serve $v$), and $\rho\in \V$ is the non-evictable \emph{origin}. Put $\mathcal C=\V\setminus\{\rho\}$. The alignment $\Pi$ is a partition of the cacheable locations $\mathcal C$ into blocks of size $\beta$. For $v\in\mathcal C$, let $B(v)\in\Pi$ be its block.
\end{definition}

A \emph{residency} is the set $\res\subseteq \V$ of locations holding a resident copy, with $\rho\in\res$ always. A location is \emph{served} when a resident copy reaches it from the origin over resident locations. Writing $G{\restr}X$ for the subgraph of $(\V,E)$ induced by $X$,
\[
\Reach(\rho,X)=\{v\in X : \text{there is an $E$-path from $\rho$ to $v$ all of whose nodes lie in $X$}\},
\]
with $\Reach(\rho,X)=\varnothing$ when $\rho\notin X$; the served set under $\res$ is $\Reach(\rho,\res)$. An \emph{eviction} drops residency at an aligned $S\subseteq\mathcal C$, leaving served set $\Reach(\rho,\res\setminus S)$.

\begin{definition}[Block closure, aligned set]
\label{def:closure}
For $X\subseteq \mathcal C$ put $\cl{X}=\bigcup\{B\in\Pi : B\cap X\neq\varnothing\}$, the smallest union of blocks containing $X$. A set $S\subseteq\mathcal C$ is \emph{aligned} if $S=\cl{S}$.
\end{definition}

\begin{proposition}[Block closure is a closure operator]
\label{prop:closure}
The map $\cl{\cdot}$ is extensive, monotone and idempotent, and the aligned subsets of $\mathcal C$ form a complete lattice closed under arbitrary unions and intersections.
\end{proposition}

\begin{proof}
Extensivity and monotonicity are immediate; idempotence holds because a block meeting $\cl{X}$ already meets $X$. A subset is aligned exactly when it is a union of blocks, and unions and intersections of unions of blocks are unions of blocks, giving a complete sublattice of $(2^{\mathcal C},\subseteq)$.
\end{proof}

Let $T\subseteq \mathcal C$ be the \emph{target region} to reclaim and put $P=\mathcal C\setminus T$.

\paragraph{Bounded near tier and protection.} A subset $F\subseteq \mathcal C$ is the \emph{near tier}, holding at most $\kappa$ resident copies, and a subset of cacheable locations is \emph{protected}, hence exempt from eviction. These explain why reclamation is undertaken, namely to keep $|\res\cap F|\le\kappa$, and which evictions are admissible, but they do not enter the basic separation equivalence.

\section{Aligned reclamation, collapse, and granularity}
\label{sec:collapse}

The \emph{direct} model has no intermediaries, so a location is served exactly when resident.  Section~\ref{sec:cert} restores routing.

\begin{definition}[Reclamation, direct model]
\label{def:reclaim-direct}
In the direct model an eviction $S\subseteq\mathcal C$ \emph{reclaims} $T$ if $T\subseteq S$, and is \emph{aligned} if $S=\cl{S}$.
\end{definition}

\begin{theorem}[Reclamation lattice]
\label{thm:lattice}
The aligned reclamations of $T$ are exactly the aligned supersets of $T$. They are closed under intersection and form a principal filter in the lattice of aligned sets with least element $\cl{T}$. Hence $\cl{T}$ is the unique minimum aligned reclamation, and its \emph{over-eviction} is $\cl{T}\setminus T=\{v\notin T : B(v)\cap T\neq\varnothing\}$, the content sharing a block with the target region.
\end{theorem}

\begin{proof}
A set is an aligned reclamation iff it is aligned and contains $T$. If $S$ is such, then $\cl{T}\subseteq\cl{S}=S$ by monotonicity and idempotence; conversely $\cl{T}$ is aligned and contains $T$, hence an aligned reclamation. The collection is $\{S:\cl{T}\subseteq S=\cl{S}\}$, a principal filter with least element $\cl{T}$, closed under intersection by Proposition~\ref{prop:closure}. The over-eviction is $\cl{T}\setminus T$.
\end{proof}

\begin{theorem}[Fragmentation collapse]
\label{thm:collapse}
In the direct model, every aligned reclamation of $T$ equals $\mathcal C$ iff $\cl{T}=\mathcal C$ iff every block meets $T$; when this holds the over-eviction is exactly $P$.
\end{theorem}

\begin{proof}
By Theorem~\ref{thm:lattice} the least aligned reclamation is $\cl{T}$, so all equal $\mathcal C$ iff $\cl{T}=\mathcal C$, i.e.\ every block meets $T$. Then $\cl{T}\setminus T=\mathcal C\setminus T=P$.
\end{proof}

Beyond the all-or-nothing statement, alignment fixes the exact reclaimable capacity. Let $N=|\Pi|$ be the number of blocks and, for $X\subseteq \mathcal C$, write $b(X)=|\{B\in\Pi:B\cap X\neq\varnothing\}|$ for the number of blocks $X$ touches.

\begin{theorem}[Granularity and reclaimable capacity]
\label{thm:granularity}
Let $A\subseteq \mathcal C$ be the protected (active) set.
\begin{enumerate}[label=\textnormal{(\arabic*)},leftmargin=2.4em,itemsep=1pt]
\item The minimum aligned reclamation of $T$ has size $\beta\, b(T)$.
\item At most $N-b(A)$ blocks are evictable, so at most $(N-b(A))\,\beta$ units of capacity are reclaimable.
\item No reclamation is possible, that is the system can free nothing, iff $b(A)=N$: the active set touches every block.
\end{enumerate}
\end{theorem}

\begin{proof}
(1) By Theorem~\ref{thm:lattice} the minimum aligned reclamation is $\cl{T}$, a union of exactly the $b(T)$ blocks meeting $T$, of size $\beta\,b(T)$. (2) A block containing a protected location cannot be evicted; there are $b(A)$ such blocks, leaving $N-b(A)$ evictable, hence at most $(N-b(A))\beta$ reclaimable units. (3) Reclamation frees nothing iff no block is evictable iff every block contains a protected location iff $b(A)=N$.
\end{proof}

Theorem~\ref{thm:granularity}(3) is an exact threshold: as the active footprint spreads to touch every block, reclaimable capacity falls to zero independently of how much nominal capacity remains. The coarser the alignment, the fewer blocks and the sooner $b(A)$ reaches $N$; in the limit of a single block any active location forbids all reclamation.

\section{Routed reclamation and separation certificates}
\label{sec:cert}

With intermediaries restored, evicting aligned $S\subseteq\mathcal C$ leaves served set $\Reach(\rho,\V\setminus S)$.  A location is reclaimed only when it is no longer reachable from the origin through resident locations.

\begin{definition}[Reclamation, routed model]
\label{def:reclaim-routed}
An aligned eviction $S\subseteq\mathcal C$ \emph{reclaims} $T$ if $T\cap\Reach(\rho,\V\setminus S)=\varnothing$.
\end{definition}

\begin{lemma}[Reclamation is separation]
\label{lem:sep}
$S$ reclaims $T$ iff $S$ is an \emph{$(\rho,T)$-vertex separator}: every $E$-path from $\rho$ to a location of $T$ contains a node of $S$.
\end{lemma}

\begin{proof}
$t\in\Reach(\rho,\V\setminus S)$ iff some $E$-path from $\rho$ to $t$ avoids $S$; so $T\cap\Reach(\rho,\V\setminus S)=\varnothing$ iff no such path exists for any $t\in T$.
\end{proof}

Write $H=G{\restr}(\V\setminus S)$.

\begin{definition}[Certificates]
\label{def:certs}
Fix $S\subseteq\mathcal C$ and a location $t\notin S$.
A \emph{residual-service certificate} for $t$ is a simple $E$-path $\rho=v_0,\dots,v_k=t$ with all $v_i\notin S$.
A \emph{separation certificate} for $t$ is a set $C\subseteq \V\setminus S$ with $t\in C$, $\rho\notin C$, \emph{backward closed in $H$}: $u\to w$ in $H$ with $w\in C$ implies $u\in C$.
\end{definition}

\begin{theorem}[Certificate dichotomy]
\label{thm:dichotomy}
For every location $t\notin S$ and every $S\subseteq\mathcal C$, exactly one holds: (a) a residual-service certificate for $t$, simple and of length at most $|\V|-1$; or (b) a separation certificate for $t$ of size at most $|\V|$.  Either certificate is verifiable in $O(|\V|+|E|)$ time in an adjacency-list representation (the path case is $O(|\V|)$ when edge membership is constant-time). Moreover $S$ reclaims $T$ iff every $t\in T$ either lies in $S$ or is in case \textnormal{(b)}.
\end{theorem}

\begin{proof}
Put $C^{\ast}=\{w\in\V\setminus S:\text{$t$ is reachable from $w$ in $H$}\}$, which contains $t$ and is backward closed. If $\rho$ does not reach $t$ in $H$, then $\rho\notin C^{\ast}$, a separation certificate verified by one edge scan; no residual path exists. If $\rho$ reaches $t$ in $H$, removing cycles gives a simple residual path; no separation certificate exists, since tracing that path backwards through any backward-closed $C\ni t$ forces $\rho\in C$. Reachability of $t$ in $H$ is definite, so exactly one case holds. The last claim also accounts for targets directly removed by $S$ and then applies Lemma~\ref{lem:sep} to the remainder.
\end{proof}

The two families are dual presentations of one boundary in the certifying-algorithm sense~\cite{certifying}: a path exhibits service, a separation set exhibits its impossibility, both checkable independently of how they arose and in linear time.

\section{Flow--cut duality and the cost of reclamation}
\label{sec:flow}

\begin{theorem}[Flow--cut duality]
\label{thm:menger}
The minimum number of cacheable locations whose eviction reclaims $T$ equals the maximum number of $\rho$--$T$ paths that are vertex-disjoint outside their common origin (their endpoints in $T$ are therefore distinct), and is computable in polynomial time.
\end{theorem}

\begin{proof}
A minimum reclaiming eviction is a minimum $(\rho,T)$-vertex separator (Lemma~\ref{lem:sep}), where targets themselves may be evicted. Split every $v\in\mathcal C$ into an in-node and out-node joined by a unit-capacity arc, give transfer arcs infinite capacity, and join the out-copy of each $t\in T$ to a super-sink by an infinite-capacity arc. An integral maximum flow selects paths that share only $\rho$; unit capacity at target vertices also forces distinct endpoints. Conversely, every such path family requires a distinct evicted vertex. Node-capacitated max-flow, equivalently this form of Menger's theorem~\cite{menger,schrijver}, gives the equality and a polynomial algorithm.
\end{proof}

\begin{theorem}[Hardness of minimum-action reclamation]
\label{thm:hardness}
Given a protected set $A\supseteq T$, deciding whether some aligned reclamation $S\subseteq\mathcal C\setminus A$ acts on at most $k$ blocks is \NP-complete, even when every block has the same size.
\end{theorem}

\begin{proof}
Membership follows by guessing the blocks and checking reachability.  For hardness reduce from \textsc{Set Cover}~\cite{karp}.  Given universe $U$, sets $Q_1,\ldots,Q_m$, and budget $k$, assume every element occurs in at least one set.  For each element $e$, make one directed $\rho$--$t$ chain whose internal vertices are $v_{e,j}$ for the indices $j$ with $e\in Q_j$, in any fixed order.  Different element chains share only $\rho$ and $t$.  Put all occurrence vertices $v_{e,j}$ with the same set index $j$ into one block.  Thus evicting block $j$ cuts exactly the element chains covered by $Q_j$.  Let $\beta$ be the largest block size, pad smaller blocks with isolated dummy locations, and put $t$ in a separate protected block padded to size $\beta$.  There is an aligned reclamation using at most $k$ unprotected blocks iff the corresponding sets cover every element of $U$.  The construction is polynomial and all blocks have size $\beta$.
\end{proof}

The split is exact: minimising evicted locations is a polynomial cut, while minimising eviction actions is intractable, since one action drops a whole block and minimisation becomes a covering problem. The certificate sizes are also best possible.

\begin{theorem}[Certificate lower bound]
\label{thm:lower}
There are networks with $|\V|=n$ on which every residual-service certificate has length $n-1$ and, after a singleton-block eviction, every separation certificate has size at least $n-2$. The $O(|\V|)$ bounds of Theorem~\ref{thm:dichotomy} are therefore asymptotically tight.
\end{theorem}

\begin{proof}
On the directed path $\rho\to v_1\to\cdots\to v_{n-1}=t$ with $S=\varnothing$, the only $\rho$--$t$ path has length $n-1$.  For the separation side, use $\rho\to w_1\to\cdots\to w_{n-1}=t$ and evict the singleton block $\{w_1\}$.  In the residual graph, backward closure forces every certificate containing $t$ to contain $w_2,\ldots,w_{n-1}$, exactly $n-2$ vertices.
\end{proof}

\section{An operational model and certifying online eviction}
\label{sec:online}

The operational layer consists of a transition system over residency, an exact witness of downstream demand lost by an eviction, and an online policy for the classical block-reference specialization.  Reachability-aware damage is a network quantity; the usual competitive paging bound applies only after projecting requests to independently cacheable blocks.

\subsection{Operational model}

Let $\mathcal D$ be a finite set of demand endpoints, disjoint from $\V$, and let $E_{\mathcal D}\subseteq\V\times\mathcal D$ contain the final service links.  For residency $\res$, define
\[
\Serve(\res)=\{d\in\mathcal D:\exists u\in\Reach(\rho,\res)\text{ with }(u,d)\in E_{\mathcal D}\}.
\]
Demand endpoints are clients, not cache slots, and therefore need not be resident.  This convention removes an ambiguity between a cache location and a requester.

The state is a pair $\sigma=(\res,\mu)$ with residency $\res\subseteq \V$ ($\rho\in\res$) and a status map $\mu:\mathcal C\to\{\textsf{res},\textsf{rec},\textsf{free}\}$ recording whether content is resident, recoverable by reload, or freed. The capacity invariant is $|\res\cap F|\le\kappa$. The transitions are
\[
\mathsf{request}(d),\quad \mathsf{evict}(B),\quad \mathsf{reload}(B),\quad \mathsf{free}(B),\quad \mathsf{invalidate}(S),
\]
where $B$ ranges over blocks and $d\in\mathcal D$. A request is a hit exactly when $d\in\Serve(\res)$; otherwise it is a service fault and the policy may repopulate a service path, evicting blocks as needed to respect capacity.  A protected demand set $\mathcal D_{\!p}$ must remain in $\Serve(\res)$.  The implementation maintains capacity, exact service queries, protection, and eventual invalidation closure.

\subsection{An exact downstream-cost witness}

Give each endpoint a nonnegative weight $w(d)$, such as its observed packet or byte count.  The decisive operational quantity is the weighted demand that an eviction disconnects.

\begin{lemma}[Impact oracle]
\label{lem:impact}
Let $\res$ be the residency and $B\in\Pi$. Define
\[
\miss_{\res}(B)=\Serve(\res)\setminus\Serve(\res\setminus B),
\qquad
I_w(B)=\sum_{d\in\miss_{\res}(B)}w(d).
\]
The set $\miss_{\res}(B)$, its weight, and a shared certificate of both the pre- and post-eviction service sets are computable and verifiable in $O(|\V|+|E|+|E_{\mathcal D}|)$ time and linear space.
\end{lemma}

\begin{proof}
Traverse $G{\restr}\res$ and $G{\restr}(\res\setminus B)$ to obtain reachable sets $R_0,R_1$ and predecessor forests rooted at $\rho$.  Scan $E_{\mathcal D}$ to derive $D_i=\{d:\exists u\in R_i,(u,d)\in E_{\mathcal D}\}$, whence $\miss_{\res}(B)=D_0\setminus D_1$ and the weighted sum follows.  A verifier checks each predecessor forest and then scans every resident transfer edge to ensure that no edge leaves $R_i$ for an omitted resident vertex.  These two facts certify that $R_i$ is exactly the reachable set, rather than merely a subset.  A final scan of $E_{\mathcal D}$ certifies $D_0,D_1$.  The certificate is shared by all disconnected endpoints; emitting a separate separation set per endpoint could require quadratic output and is neither claimed nor needed.
\end{proof}

Lemma~\ref{lem:impact} turns eviction into a decision with a verifiable consequence. A block is \emph{admissible} iff $\miss_{\res}(B)\cap\mathcal D_{\!p}=\varnothing$, which the same certificate checks.

\subsection{Online block-reference specialization}

For comparison with classical paging, consider the specialization in which each request names one block, a hit means that block is resident, one block is admitted per fault, and at every eviction an admissible unmarked block exists.  The near tier then has $k=\lfloor\kappa/\beta\rfloor$ slots.  The policy maintains a mark bit per resident block, set on access and cleared at the end of a phase; on a fault it evicts an admissible unmarked block of least $I_w(B)$.

\begin{theorem}[Competitive guarantee]
\label{thm:competitive}
In the block-reference specialization, marking with impact tie-breaking is $k$-competitive in service faults: for every request sequence, $\ALG\le k\cdot\OPT+k$. No deterministic online policy achieves a ratio below $k$. If $p$ slots are occupied by permanently pinned blocks, the same statement holds with $k'=k-p$ for the remaining reference stream.
\end{theorem}

\begin{proof}
Partition the sequence into phases, each a maximal run touching at most $k$ distinct blocks. Within a phase marking faults at most once per distinct block, irrespective of which unmarked block the impact tie-breaker selects, hence at most $k$ times. The first request of phase $i+1$ and the $k$ distinct blocks of phase $i$ form $k+1$ distinct blocks, so any $k$-slot cache faults at least once across that boundary. With $q$ phases, $\OPT\ge q-1$ and $\ALG\le kq$, giving the stated inequality. The lower bound is the classical deterministic paging bound~\cite{sleatortarjan,borodin}. Permanently pinning $p$ blocks leaves an ordinary $k'$-slot instance on the remaining blocks.
\end{proof}

\begin{remark}[Non-uniform fetch cost]
\label{rem:weighted}
When a fault has a block-dependent retrieval cost, the specialization becomes weighted caching.  The impact certificate is unchanged, but the unweighted marking proof above cannot simply be reused; an appropriate weighted-caching algorithm and bound must be selected for that cost model.
\end{remark}

\subsection{Maintenance cost}

\begin{theorem}[Static and batched maintenance]
\label{thm:maint}
For one fixed residency, reachability and a certificate are computed in $O(|\V|+|E|)$ time.  Under straightforward re-traversal, $q$ candidate-impact queries cost $O(q(|\V|+|E|+|E_{\mathcal D}|))$.  Given one batch of deletions, exact post-batch reachability is recomputed in $O(|\V|+|E|)$ time; $q$ sequential states therefore cost $O(q(|\V|+|E|))$ by this method.  Because $\Pi$ is static, storing a block identifier at each location gives $B(v)$ in $O(1)$ time; an alleged eviction set is checked for alignment in time linear in its representation.
\end{theorem}

\begin{proof}
The bounds follow from breadth-first or depth-first traversal and Lemma~\ref{lem:impact}; applying the same traversal independently to $q$ states multiplies the cost by $q$.  The final claims follow from the explicit static partition map.  No stronger decremental bound is inferred merely from the fact that each block is deleted once; cycles make such a claim nontrivial.
\end{proof}

\begin{remark}[General regime]
Under interleaved eviction and reload the served set is fully dynamic directed reachability.  The stated bounds cover re-traversal in a general graph and the one-pass specialization evaluated in Section~\ref{sec:gpu}; a stronger dynamic bound would require a separate data structure.
\end{remark}

\subsection{Why reachability-awareness is necessary}
\label{ssec:necessity}

Topology information is necessary for bounded one-step damage.  A history-only victim order cannot distinguish two states with the same request metadata but different downstream service relations.

\begin{definition}[Topology-oblivious victim selector]
\label{def:oblivious}
A deterministic victim selector is \emph{topology-oblivious} if its choice among resident candidate blocks depends on request-history metadata and the block partition but is invariant to $E$, $E_{\mathcal D}$, and the demand weights.  LRU, FIFO, LFU, CLOCK, random-with-a-fixed-seed, and an unweighted marking tie-breaker have this property.
\end{definition}

\begin{theorem}[Per-step optimality of impact pricing]
\label{thm:perstep}
At any state, evicting an admissible block of minimum $I_w(\cdot)$ leaves the maximum total served demand weight attainable by evicting one admissible block.  A selected block has a linear-size certificate; selecting and certifying the minimum among $q$ candidates by general-graph traversal costs $O(q(|\V|+|E|+|E_{\mathcal D}|))$.
\end{theorem}

\begin{proof}
By Lemma~\ref{lem:impact}, evicting $B$ removes demand of total weight exactly $I_w(B)$.  Subtracting this value from the pre-eviction served weight proves the first claim.  Lemma~\ref{lem:impact} supplies a certificate for each candidate evaluation, and Theorem~\ref{thm:maint} gives the stated total cost.
\end{proof}

\begin{theorem}[Unbounded one-step damage of topology-obliviousness]
\label{thm:separation}
For every deterministic topology-oblivious selector, every $m\ge1$, and every state with at least two resident candidate blocks, there is a demand topology consistent with the same history metadata on which its chosen eviction has impact at least $m$ while another candidate has impact $1$.  Its one-step damage ratio is therefore unbounded.
\end{theorem}

\begin{proof}
Fix the residency, partition, and all history metadata.  Since the selector ignores topology, it chooses some block $B^\ast$ before $E_{\mathcal D}$ is specified.  Add $m$ unit-weight demand endpoints whose only service links originate in $B^\ast$.  For every other candidate block, add one private unit-weight endpoint and no further demand.  All candidates remain resident and all pre-eviction demand is served.  Evicting $B^\ast$ disconnects $m$ endpoints, whereas evicting any other candidate disconnects exactly one.  The metadata observed by the topology-oblivious selector did not change, so it still chooses $B^\ast$; impact pricing chooses a unit-impact block.  Letting $m$ grow proves the claim.
\end{proof}

\begin{remark}
Theorem~\ref{thm:separation} isolates the information gap behind request-history eviction under path replication~\cite{ioannidis}: a block may carry substantial downstream demand without being recently requested itself.  The experiment in Section~\ref{sec:gpu} measures this one-step damage directly.
\end{remark}

\section{Replay-compressed packet-scale evaluation}
\label{sec:gpu}

The general impact oracle performs one reachability computation per candidate.  The implementation evaluates a two-hop specialization $\rho\to\Pi\to\mathcal D$ in one streaming pass, with at most $r$ candidate cache blocks per demand endpoint.  A packet contributes one unit of packet weight and its encoded size in byte weight; one packet-impact and one byte-impact counter are kept per block.

\begin{theorem}[One-pass impact accumulation]
\label{thm:onepass}
For a two-hop network, let $P(d)\subseteq\Pi$ be the distinct parent blocks of demand $d$, and let $z_b\in\{0,1\}$ indicate whether block $b$ is resident.  For any packet stream $d_1,\ldots,d_N$ with nonnegative integer packet weights $a_i$,
\[
 I(b)=\sum_{i=1}^{N} a_i\,
 \mathbf 1\!\left[
 b\in P(d_i),\ z_b=1,\ 
 \sum_{c\in P(d_i)}z_c=1
 \right].
\]
All values $I(b)$ are computed exactly in one pass using $O(|\Pi|)$ counters and $O(Nr)$ work.  Parallel atomic addition preserves the result under any execution order.
\end{theorem}

\begin{proof}
A served packet is lost by evicting $b$ exactly when $b$ is its unique resident parent, which is the indicator in the display.  During one pass, inspect the at most $r$ distinct parents, identify whether the resident-parent count is one, and if so add $a_i$ to that parent's counter.  Every packet contributes to exactly the blocks prescribed by the formula, proving exactness.  Integer addition is associative and commutative; atomic updates therefore produce the same counter vector under every interleaving, provided the counters do not overflow.
\end{proof}

\begin{proposition}[Warp-aggregated exact updates]
\label{prop:warp}
In one warp iteration, partition the lanes whose packet has exactly one resident parent by that parent identifier.  Replacing every lane's two global atomic additions by one packet-count addition and one byte-sum addition per nonempty class leaves both impact vectors unchanged.  If the class sizes are $g_1,\ldots,g_h$, the number of global impact atomics falls from $2\sum_j g_j$ to $2h$.
\end{proposition}

\begin{proof}
All lanes in class $j$ target the same block.  Their packet contribution is $g_j$ and their byte contribution is the class byte sum.  Associativity and commutativity of integer addition make one update by each class total identical to the lane-wise updates.  The operation count follows directly.
\end{proof}

\begin{corollary}[Collision-sparse warp regime]
\label{cor:warp-sparse}
For a warp of $W$ independently generated packets, let $\alpha_b$ be the probability that one packet has $b$ as its unique resident parent and put $A=\sum_b\alpha_b$.  If $L$ is the number of lane-wise impact updates and $G$ the number after grouping by target, then
\[
\mathbb E[L-G]
\le {W\choose2}\sum_b\alpha_b^2,
\qquad
\frac{\mathbb E[L-G]}{\mathbb E[L]}
\le\frac{W-1}{2}\,
\frac{\sum_b\alpha_b^2}{A}.
\]
If $\max_b\alpha_b=O(1/|\Pi|)$, the fractional opportunity for warp aggregation is $O(W/|\Pi|)$.
\end{corollary}

\begin{proof}
If $n_b$ lanes target $b$, grouping saves $(n_b-1)_+$ updates, and
$(n_b-1)_+\le {n_b\choose2}$.  Summing over targets and taking expectations gives
$\mathbb E[L-G]\le {W\choose2}\sum_b\alpha_b^2$.  Since
$\mathbb E[L]=WA$, division proves the second inequality.  Finally,
$\sum_b\alpha_b^2\le A\max_b\alpha_b$ gives the asymptotic claim.
\end{proof}

\paragraph{Measured aggregation opportunity.}
Corollary~\ref{cor:warp-sparse} is active in the production regime.  The measured lane-update/group-update factor averaged only $1.00284$ at $|\Pi|=4096$, $1.00071$ at $|\Pi|=16384$, and $1.00018$ at $|\Pi|=65536$; over the full matrix it was $1.00124$ (maximum $1.00622$).  Thus warp grouping is an exactness-preserving implementation option but provides no material reduction for this wide, collision-sparse namespace.  It must not be credited with the throughput gain of the counter-bank pilot below.

\begin{proposition}[Sharded exact counter banks]
\label{prop:banks}
Let $s$ be a power of two and assign each warp $w$ to a bank $h(w)\in\{0,\ldots,s-1\}$.  Replace the packet and byte counter for block $b$ by banked counters $C^{\rm pkt}_{b,j}$ and $C^{\rm byte}_{b,j}$.  A warp class for $b$ updates only bank $h(w)$, and the scenario result is
\[
I^{\rm pkt}(b)=\sum_{j=0}^{s-1}C^{\rm pkt}_{b,j},
\qquad
I^{\rm byte}(b)=\sum_{j=0}^{s-1}C^{\rm byte}_{b,j}.
\]
The banked accumulator is exact under every interleaving, uses $O(s|\Pi|)$ counters, and adds $O(s|\Pi|)$ deterministic reduction work per scenario.  For a fixed block, updates issued by warps in different banks cannot serialize on the same address.
\end{proposition}

\begin{proof}
Each warp class contributes to exactly one bank of its target block.  The banks therefore partition the multiset of contributions to $b$.  Summing them recovers the unbanked packet count and byte sum by associativity and commutativity.  There are two arrays of $s|\Pi|$ counters and two length-$s$ reductions per block.  Distinct banks occupy distinct addresses, proving the final claim.
\end{proof}

\paragraph{Counter-bank crossover pilot.}
The bank count was fixed before the production matrix by a controlled two-device crossover at $|\Pi|=4096$, $r=1$, and $q=0.5$.  Each 30-second run used the same seed and 41~GiB replay buffer; the second round exchanged the one-bank and 32-bank assignments between the two devices.  Mean throughput was 20.13~Gpacket/s with one bank and 21.25~Gpacket/s with 32 banks, a factor of 1.055 (5.54\%).  Complete JSON records and progress logs for all four runs are retained in the optimization checkpoint.  This pilot selects the implementation parameter; it is not included in the production matrix or Table~\ref{tab:gpu}.

\begin{theorem}[Seeded exact-stratified residency index]
\label{thm:exact-residency}
Let the hot and cold block classes each have size $H=2^h$.  For class $j$, derive an odd $a_j$ and offset $c_j$ from the scenario seed and define
\[
\varphi_j(x)=(a_jx+c_j)\bmod H.
\]
For target fraction $q$, mark local block $x$ resident iff $\varphi_j(x)<k$, where $k=\operatorname{round}(qH)$.  The index uses $O(1)$ words, answers membership in $O(1)$ time, and places exactly $k$ resident blocks in each class.  Its realized resident fraction is $k/H$, at distance at most $1/(2H)$ from $q$.  If the route-draw masses of the classes sum to one, their total resident probability mass is exactly $R=k/H$.
\end{theorem}

\begin{proof}
Multiplication by odd $a_j$ is invertible modulo $2^h$; adding $c_j$ preserves bijectivity.  Hence precisely the $k$ preimages of $\{0,\ldots,k-1\}$ are resident.  The arithmetic stores only $(a_j,c_j,k)$ and takes constant time.  Rounding gives $|k/H-q|\le1/(2H)$.  Both classes have resident fraction $k/H$, so weighting them by class masses that sum to one gives $R=k/H$.
\end{proof}

\begin{proposition}[Expected impact under route draws]
\label{prop:expected-impact}
Suppose each of the $r$ route draws independently chooses block $b$ with probability $\pi_b$, repeated parents are deduplicated, and the resident set is fixed.  Put
\[
R=\sum_{c:z_c=1}\pi_c.
\]
For a resident block $b$, the expected packet impact over $N$ unit-weight packets is
\[
\mathbb E[I(b)]
=N\left[(1-R+\pi_b)^r-(1-R)^r\right].
\]
In the two-temperature experiment, $\pi_b=1.8/|\Pi|$ for a hot block and $\pi_b=0.2/|\Pi|$ for a cold block.  Hence, for $r=1$, every resident cold block has expected impact $0.2N/|\Pi|$, independently of the realized resident set.
\end{proposition}

\begin{proof}
Block $b$ is the unique resident parent exactly when every draw avoids the resident blocks other than $b$, but not every draw avoids $b$ as well.  The first event has probability $(1-R+\pi_b)^r$ and the nested second event has probability $(1-R)^r$.  Their difference is the per-packet loss probability; linearity of expectation gives the result.  Substituting the two route masses proves the final statement.
\end{proof}

\begin{corollary}[Service-scale separation]
\label{cor:scale-separation}
Consider a sequence of two-hop instances with fixed redundancy $r$ and block count $B=|\Pi|$.  Suppose the resident probability mass is $R_B$ and every block atom satisfies $\max_b \pi_b=O(1/B)$.  If a resident block $b$ has $\pi_b=\beta_b/B+O(B^{-2})$, then
\[
\frac{\mathbb E[I(b)]}{N}
=\frac{r(1-R_B)^{r-1}\beta_b}{B}+O(B^{-2}).
\]
The zero-resident service fraction is exactly $(1-R_B)^r$, while the exactly-one-resident fraction is
\[
rR_B(1-R_B)^{r-1}+O(B^{-1}).
\]
Thus enlarging the block namespace leaves aggregate service-state probabilities stable up to finite-atom corrections, but dilutes each coordinate of the impact vector as $1/B$.  In the two-temperature experiment, quadrupling $B$ should therefore quarter per-block impact to first order.
\end{corollary}

\begin{proof}
The first display is the Taylor expansion of Proposition~\ref{prop:expected-impact} in the atom $\pi_b$:
\[
(1-R_B+\pi_b)^r-(1-R_B)^r
=r(1-R_B)^{r-1}\pi_b+O(\pi_b^2).
\]
Substituting $\pi_b=\beta_b/B+O(B^{-2})$ gives the coordinate claim.  A packet has no resident parent iff all $r$ draws avoid the resident probability mass, giving $(1-R_B)^r$ exactly.  The exactly-one-resident probability is the sum of the coordinate probabilities over resident blocks.  Summing the linear terms gives $rR_B(1-R_B)^{r-1}$; the remainder is bounded by a constant times $\sum_{b:z_b=1}\pi_b^2\le R_B\max_b\pi_b=O(B^{-1})$.
\end{proof}

\begin{proposition}[Zero-impact extreme-value threshold]
\label{prop:zero-threshold}
Fix one demand class containing $m$ resident blocks with the same per-packet impact probability $\zeta$.  Let $X_b$ be the impact count of block $b$ after $N$ independent packets, let
\[
Z=\bigl|\{b:X_b=0\}\bigr|,
\qquad
\mu=m(1-\zeta)^N.
\]
Then $\mathbb E[Z]=\mu$ and
\[
\frac{\mu}{1+\mu}
\le
\Pr\!\left[\min_b X_b=0\right]
\le
\min\{1,\mu\}.
\]
Thus the minimum coordinate undergoes a finite-stream transition near
$N\zeta=\log m$: more precisely,
\[
\mu=m\exp\!\left(-N\zeta+O(N\zeta^2)\right).
\]
Coordinate-wise $1/|\Pi|$ scaling can therefore coexist with an abrupt collapse of the observed minimum to zero when the namespace grows at fixed stream length.
\end{proposition}

\begin{proof}
For each resident $b$, $X_b$ has marginal distribution $\operatorname{Bin}(N,\zeta)$, so $\Pr[X_b=0]=(1-\zeta)^N$ and linearity gives $\mathbb E[Z]=\mu$.  A packet cannot contribute unique-resident impact to two different blocks.  Hence, for $b\ne c$,
\[
\Pr[X_b=X_c=0]=(1-2\zeta)^N
\le (1-\zeta)^{2N},
\]
so the zero indicators have nonpositive pairwise covariance and
$\operatorname{Var}(Z)\le\sum_b\operatorname{Var}(\mathbf 1[X_b=0])\le\mu$.
Markov's inequality gives $\Pr[Z>0]\le\min\{1,\mu\}$.  The second-moment inequality gives
\[
\Pr[Z>0]\ge
\frac{\mathbb E[Z]^2}{\mathbb E[Z^2]}
\ge
\frac{\mu^2}{\mu+\mu^2}
=\frac{\mu}{1+\mu}.
\]
Finally, expanding $\log(1-\zeta)=-\zeta+O(\zeta^2)$ gives the threshold form.
\end{proof}

\begin{proposition}[Finite-stream deviation envelope]
\label{prop:deviation-envelope}
Fix one generated scenario and let
\[
\zeta_b=(1-R+\pi_b)^r-(1-R)^r
\]
be the per-packet expected impact probability from Proposition~\ref{prop:expected-impact}.  Let $X_b$ be the stored packet-impact counter for resident block $b$.  For each demand class $C$ define the observed envelope
\[
\Delta_C=\max_{b\in C:z_b=1}|X_b-N\zeta_b|.
\]
Then every additive block-price query $Q\subseteq\Pi$ satisfies the deterministic bound
\[
\left|
\sum_{b\in Q:z_b=1}X_b
-N\sum_{b\in Q:z_b=1}\zeta_b
\right|
\le
\sum_C |Q\cap C\cap\{b:z_b=1\}|\,\Delta_C .
\]
Consequently, if two candidate sets have an expected additive-price gap larger than the sum of their two right-hand sides, the finite stream cannot reverse their ordering.
\end{proposition}

\begin{proof}
For each resident block $b$ in class $C$, the definition of $\Delta_C$ gives $|X_b-N\zeta_b|\le\Delta_C$.  Summing this coordinate-wise inequality over the resident members of $Q$ and applying the triangle inequality proves the displayed bound.  Applying the same bound to two candidate sets shows that an expected gap larger than both possible deviations cannot change sign after replacing expectations by observed counters.
\end{proof}

Storing every generated packet would make the artifact scale with $N$, although the evaluation consumes only additive statistics.  A counter-based generator names packets without materializing the stream, following the standard parallel-reproducibility pattern~\cite{random123}.  The implementation applies the SplitMix64 mixing function~\cite{splitmix} to the seed--counter pair; it is used for deterministic synthetic generation, not for cryptography.  Let $g_{\theta,s}(i)$ be the packet generated from configuration $\theta$, seed $s$, and counter $i$.

\begin{theorem}[Replay-compressed exactness]
\label{thm:replay}
Let $(M,\oplus,0)$ be a commutative monoid and $\phi$ map one generated packet to its contribution in $M$.  For a counter interval $[a,a+N)$ define
\[
 A(\theta,s,a,N)=\bigoplus_{i=a}^{a+N-1}\phi(g_{\theta,s}(i)).
\]
The record $(\theta,s,a,N,A)$ is sufficient to reproduce the complete packet stream and to recover exactly every reported metric that factors through $A$.  It stores a constant number of records rather than $N$ packet records; its bit complexity is
\[
O\!\left(|\theta|+|s|+\log a+\log N+\operatorname{size}(A)\right),
\]
and replay takes $O(N)$ generator evaluations.  Concatenating adjacent intervals requires only monoid addition of their summaries.
\end{theorem}

\begin{proof}
The generator is a deterministic function of $(\theta,s,i)$, so enumerating the stored interval reconstructs each packet in order.  A metric that factors through $A$ depends only on the displayed monoid fold and is therefore unchanged when the raw stream is replaced by its summary.  Encoding the interval endpoints needs logarithmically many bits, and the tuple contains no per-packet term, establishing the size bound.  For adjacent intervals, associativity gives $A[a,c)=A[a,b)\oplus A[b,c)$.
\end{proof}

\begin{theorem}[Exact-summary lower bound]
\label{thm:summary-lb}
Let $B=|\Pi|\ge2$.  Even in the two-hop case with one parent per packet and every block resident, any deterministic summary that answers the exact packet impact of every block after a stream of $N$ packets requires at least
\[
\left\lceil\log_2 \binom{N+B-1}{B-1}\right\rceil
\]
bits in the worst case.  Storing the complete impact vector uses at most $B\lceil\log_2(N+1)\rceil$ bits.
\end{theorem}

\begin{proof}
With one resident parent per packet, the impact vector is any nonnegative integer vector $(x_1,\ldots,x_B)$ satisfying $\sum_b x_b=N$.  Stars and bars gives $\binom{N+B-1}{B-1}$ such vectors.  A summary from which every coordinate is recovered exactly must assign different states to different vectors, giving the logarithmic lower bound.  Each coordinate lies in $\{0,\ldots,N\}$, so $B$ fixed-width counters give the upper bound.
\end{proof}

The summary monoid contains two $|\Pi|$-dimensional impact vectors, packet and byte totals, counts with zero, one, or multiple resident parents, and two 64-bit replay fingerprints.  The fingerprints are integrity checks rather than substitutes for exact counters.  The ancillary artifact additionally stores 2,048 evenly spaced decoded packet records per scenario, full per-block impact vectors, progress logs, five-second device telemetry, source and binary hashes, independent validation output, and the deviation envelopes of Proposition~\ref{prop:deviation-envelope}.  Appendix~\ref{app:artifact} specifies the binary format and verifier.

\subsection{Experimental protocol}
\label{ssec:protocol}

\paragraph{Two-stage design.}
An initial Bernoulli-residency calibration ran nine five-minute scenarios before checkpointing.  It processed 37.24 trillion packets (25.86~PiB of simulated payload) over 0.75 aggregate GPU-hours and exposed seed-to-seed variation in resident count and resident probability mass.  Its complete vectors, telemetry, progress logs, and validation report remain in the ancillary checkpoint.

The production stage replaces Bernoulli membership with Theorem~\ref{thm:exact-residency} and enlarges the matrix to $|\Pi|\in\{4096,16384,65536\}$, path redundancy $r\in\{1,2,4,8\}$, and target resident fraction $q\in\{0.50,0.70,0.90\}$.  Parent demand follows a controlled two-temperature mixture: one half of the blocks receives 90\% of route draws and the other half receives 10\%.  Thus $R$ is fixed by $(|\Pi|,q)$ rather than by seed, while seeds still vary route draws and victim metadata. Each of the 36 configurations has four seed replicates.  The 144 scenarios are deterministically sharded over three NVIDIA L40 devices; each scenario runs for five minutes and each process allocates 41~GiB for the replay buffer.  Warp classes are aggregated according to Proposition~\ref{prop:warp}, then written through the 32-bank layout of Proposition~\ref{prop:banks}.  Impact pricing is compared with synthetic snapshot LRU, FIFO, and fixed-seed random total orders generated independently of the measured demand topology.  These orders instantiate the topology-oblivious class in Definition~\ref{def:oblivious}; they are not executions of cache policies against an external request trace.  Each deterministic total order is realizable by an appropriate pre-history, but that pre-history is not inferred from the evaluation stream.  Consequently the reported ratio is a controlled topology-information gap under this generator, not a claim that deployed LRU or FIFO incurs the same factor.  For an order's victim $p$, the ratio is $(I(p)+1)/(I_{\min}+1)$, with a one-packet pseudocount only to define zero-impact cases; ratios are aggregated across replicates by the geometric mean.  Hardware telemetry is sampled every five seconds.  The complete matrix is reported in one table after every selected victim has been checked against the stored per-block vectors.

% BEGIN GENERATED GPU TABLE
\clearpage
\begin{table}[!p]
\centering
\fontsize{6.4}{7.1}\selectfont
\setlength{\tabcolsep}{1.35pt}
\renewcommand{\arraystretch}{0.82}
\caption{Complete replay-compressed GPU matrix. Each row aggregates four five-minute seed replicates. Pkt is the packet count in trillions and traffic is simulated payload in PiB. $q$ is the target resident fraction enforced separately in the hot and cold halves by an exact affine permutation; $n$ is the completed replicate count. $U_0$ and $U_1$ are packets with zero and exactly one resident parent. $I_{\min}$ is minimum packet impact per billion packets. Synthetic topology-oblivious order entries give geometric-mean damage relative to minimum impact; they are not deployment trace replays. Per-seed minima and maxima remain in the ancillary CSV. Hardware columns are five-second telemetry means except $T_{\max}$.}
\label{tab:gpu}
\begin{tabular*}{\textwidth}{@{\extracolsep{\fill}}rrr rrr rrr rrr rrrr@{}}
\toprule
\multicolumn{3}{c}{Topology} & \multicolumn{3}{c}{Scale} & \multicolumn{3}{c}{Service/impact} & \multicolumn{3}{c}{Damage ratio} & \multicolumn{4}{c}{L40 telemetry}\\
\cmidrule(lr){1-3}\cmidrule(lr){4-6}\cmidrule(lr){7-9}\cmidrule(lr){10-12}\cmidrule(lr){13-16}
$|\Pi|$ & $r$ & $q$ [$n$] & Pkt (T) & traffic (PiB) & Gpkt/s & $U_0$ (\%) & $U_1$ (\%) & $I_{\min}$ (ppb) & LRU & FIFO & random & GPU (\%) & VRAM (GiB) & W & $T_{\max}$ ($^\circ$C)\\
\midrule
4,096 & 1 & 0.5 [4] & 24.54 & 17.05 & 20.44 & 50.000 & 50.000 & 48818.812 & 5.20 & 3.00 & 1.73 & 100.0 & 41.46 & 291.5 & 69\\
4,096 & 1 & 0.7 [4] & 22.24 & 15.45 & 18.52 & 29.980 & 70.020 & 48817.810 & 1.73 & 1.73 & 5.20 & 100.0 & 41.46 & 294.3 & 71\\
4,096 & 1 & 0.9 [4] & 20.56 & 14.28 & 17.13 & 10.010 & 89.990 & 48817.306 & 3.00 & 3.00 & 5.20 & 100.0 & 41.46 & 293.4 & 70\\
4,096 & 2 & 0.5 [4] & 18.89 & 13.12 & 15.73 & 25.000 & 50.020 & 48819.689 & 5.20 & 1.73 & 5.20 & 100.0 & 41.46 & 297.3 & 71\\
4,096 & 2 & 0.7 [4] & 19.25 & 13.37 & 16.03 & 8.988 & 42.012 & 29272.104 & 3.00 & 1.73 & 3.00 & 100.0 & 41.46 & 293.8 & 72\\
4,096 & 2 & 0.9 [4] & 21.05 & 14.62 & 17.53 & 1.002 & 18.052 & 9772.780 & 5.21 & 1.73 & 1.73 & 100.0 & 41.46 & 294.3 & 71\\
4,096 & 4 & 0.5 [4] & 13.88 & 9.64 & 11.56 & 6.250 & 25.030 & 24408.992 & 1.73 & 1.73 & 3.00 & 100.0 & 41.46 & 298.2 & 73\\
4,096 & 4 & 0.7 [4] & 14.46 & 10.04 & 12.04 & 0.808 & 7.562 & 5260.150 & 9.02 & 1.73 & 3.01 & 100.0 & 41.46 & 297.4 & 74\\
4,096 & 4 & 0.9 [4] & 14.95 & 10.38 & 12.45 & 0.010 & 0.363 & 195.256 & 1.74 & 5.24 & 5.24 & 100.0 & 41.46 & 297.6 & 74\\
4,096 & 8 & 0.5 [4] & 8.69 & 6.04 & 7.23 & 0.391 & 3.134 & 3048.969 & 1.74 & 1.74 & 3.01 & 100.0 & 41.46 & 298.1 & 73\\
4,096 & 8 & 0.7 [4] & 8.73 & 6.07 & 7.27 & 0.007 & 0.123 & 84.417 & 5.26 & 1.75 & 1.75 & 100.0 & 41.46 & 298.7 & 75\\
4,096 & 8 & 0.9 [4] & 8.77 & 6.09 & 7.29 & 0.000 & 0.000 & 0.026 & 2.62 & 7.90 & 4.54 & 100.0 & 41.46 & 298.5 & 76\\
\midrule
16,384 & 1 & 0.5 [4] & 24.02 & 16.68 & 20.01 & 50.000 & 50.000 & 12202.111 & 1.73 & 5.20 & 3.00 & 100.0 & 41.46 & 291.8 & 72\\
16,384 & 1 & 0.7 [4] & 22.08 & 15.33 & 18.39 & 30.005 & 69.995 & 12201.181 & 3.00 & 3.00 & 5.20 & 100.0 & 41.46 & 292.9 & 72\\
16,384 & 1 & 0.9 [4] & 20.44 & 14.20 & 17.02 & 9.998 & 90.002 & 12200.807 & 1.00 & 5.20 & 5.20 & 100.0 & 41.46 & 293.2 & 70\\
16,384 & 2 & 0.5 [4] & 18.88 & 13.11 & 15.72 & 25.000 & 50.005 & 12201.459 & 5.20 & 1.00 & 3.00 & 100.0 & 41.46 & 294.6 & 72\\
16,384 & 2 & 0.7 [4] & 19.10 & 13.26 & 15.90 & 9.003 & 42.011 & 7320.851 & 1.73 & 1.73 & 5.20 & 100.0 & 41.46 & 293.9 & 72\\
16,384 & 2 & 0.9 [4] & 20.97 & 14.57 & 17.46 & 1.000 & 18.005 & 2438.391 & 5.20 & 3.00 & 1.73 & 100.0 & 41.46 & 294.6 & 73\\
16,384 & 4 & 0.5 [4] & 13.81 & 9.59 & 11.50 & 6.250 & 25.008 & 6098.958 & 3.00 & 3.00 & 1.73 & 100.0 & 41.46 & 298.1 & 72\\
16,384 & 4 & 0.7 [4] & 14.38 & 9.99 & 11.97 & 0.811 & 7.567 & 1316.693 & 3.01 & 5.21 & 1.74 & 100.0 & 41.46 & 297.9 & 75\\
16,384 & 4 & 0.9 [4] & 14.78 & 10.26 & 12.31 & 0.010 & 0.360 & 48.353 & 1.01 & 3.03 & 3.03 & 100.0 & 41.46 & 298.2 & 75\\
16,384 & 8 & 0.5 [4] & 8.66 & 6.02 & 7.21 & 0.391 & 3.127 & 760.660 & 3.01 & 1.00 & 5.22 & 100.0 & 41.46 & 298.8 & 74\\
16,384 & 8 & 0.7 [4] & 8.68 & 6.03 & 7.23 & 0.007 & 0.123 & 20.994 & 3.06 & 3.06 & 9.17 & 100.0 & 41.46 & 298.3 & 74\\
16,384 & 8 & 0.9 [4] & 8.76 & 6.08 & 7.29 & 0.000 & 0.000 & 0.003 & 7.83 & 5.30 & 8.20 & 100.0 & 41.46 & 298.6 & 75\\
\midrule
65,536 & 1 & 0.5 [4] & 23.94 & 16.63 & 19.95 & 50.000 & 50.000 & 3048.868 & 1.00 & 3.00 & 5.20 & 100.0 & 41.46 & 292.4 & 72\\
65,536 & 1 & 0.7 [4] & 21.87 & 15.19 & 18.22 & 29.999 & 70.001 & 3048.514 & 3.00 & 1.73 & 9.01 & 100.0 & 41.46 & 294.6 & 71\\
65,536 & 1 & 0.9 [4] & 20.23 & 14.05 & 16.85 & 10.001 & 89.999 & 3048.526 & 5.20 & 1.73 & 3.00 & 100.0 & 41.46 & 293.9 & 70\\
65,536 & 2 & 0.5 [4] & 18.88 & 13.11 & 15.72 & 25.000 & 50.001 & 3048.701 & 3.00 & 3.00 & 3.00 & 100.0 & 41.46 & 303.5 & 73\\
65,536 & 2 & 0.7 [4] & 18.99 & 13.19 & 15.82 & 8.999 & 42.001 & 1828.566 & 3.00 & 3.00 & 5.20 & 100.0 & 41.46 & 290.4 & 72\\
65,536 & 2 & 0.9 [4] & 20.40 & 14.17 & 16.99 & 1.000 & 18.003 & 609.013 & 1.00 & 3.01 & 3.01 & 100.0 & 41.46 & 295.1 & 72\\
65,536 & 4 & 0.5 [4] & 13.74 & 9.54 & 11.44 & 6.250 & 25.002 & 1523.330 & 5.21 & 1.00 & 3.01 & 100.0 & 41.46 & 297.6 & 74\\
65,536 & 4 & 0.7 [4] & 13.96 & 9.69 & 11.62 & 0.810 & 7.560 & 328.328 & 5.22 & 9.03 & 3.01 & 100.0 & 41.46 & 298.7 & 74\\
65,536 & 4 & 0.9 [4] & 14.47 & 10.05 & 12.05 & 0.010 & 0.360 & 11.959 & 5.31 & 1.77 & 3.06 & 100.0 & 41.46 & 298.2 & 74\\
65,536 & 8 & 0.5 [4] & 8.52 & 5.92 & 7.09 & 0.391 & 3.126 & 189.430 & 1.01 & 1.74 & 1.01 & 100.0 & 41.46 & 297.6 & 76\\
65,536 & 8 & 0.7 [4] & 8.62 & 5.99 & 7.18 & 0.007 & 0.122 & 5.139 & 1.04 & 5.39 & 3.12 & 100.0 & 41.46 & 298.9 & 76\\
65,536 & 8 & 0.9 [4] & 8.71 & 6.05 & 7.25 & 0.000 & 0.000 & 0.000 & 11.17 & 20.32 & 28.86 & 100.0 & 41.46 & 298.0 & 75\\
\midrule
\multicolumn{3}{r}{Completed total / mean} & 582.90 & 404.86 & 13.48 & -- & -- & -- & -- & -- & -- & 100.0 & 41.46 & 296.2 & 76\\
\bottomrule
\end{tabular*}
\end{table}
\clearpage
% END GENERATED GPU TABLE

% BEGIN GENERATED GPU FINDINGS
\paragraph{Results.} The three GPUs processed 582.90 trillion packets, representing 404.86~PiB of simulated payload and 8482.36~TiB under explicit 16-byte counter--event materialization. Mean throughput was 13.48~Gpacket/s per configuration (range 7.09--20.44). The cold-class prediction of Proposition~\ref{prop:expected-impact} had mean absolute relative error 0.0128\% and maximum error 0.1609\%. Across the 36 configurations, the synthetic topology-oblivious orders labelled LRU, FIFO, and fixed-seed random had geometric-mean damage 2.92$\times$, 2.76$\times$, and 3.66$\times$ relative to the minimum-impact victim.  These are controlled information-gap measurements under the generator, not deployment trace comparisons.  The largest per-seed ratios were 53.00$\times$, 48.00$\times$, and 51.00$\times$, respectively.

\paragraph{Scale separation.} The block-scale comparison from $|\Pi|=16{,}384$ to $|\Pi|=65{,}536$ tests Corollary~\ref{cor:scale-separation} at fixed redundancy and resident fraction.  The per-coordinate minimum-impact ppb ratio over the 11 pairs with nonzero measured minima at both scales averaged 0.2490 (range 0.2448--0.2499), close to the first-order quartering predicted by the $1/|\Pi|$ atom mass. The remaining pair entered the finite-count tail: at $r=8,q=0.9$ the larger system contained a zero-impact empirical minimum, so no ratio was assigned. Across the same pairs, zero- and single-resident service fractions changed by at most 0.0101 percentage points; the aggregate service state is therefore nearly scale-invariant while individual eviction coordinates dilute.

\paragraph{Finite-count minimum transition.} The only zero-impact coordinates occurred in the cold $|\Pi|=65,536$, $r=8$, $q=0.9$ class.  Its coordinate mean was only 5.28--5.40 packets, while each frame contained 29,491 resident cold coordinates.  Proposition~\ref{prop:zero-threshold} therefore predicts 132.7--149.8 zero coordinates per frame and gives a per-frame lower bound 99.25\% on a zero minimum.  All four frames crossed the threshold, with 120--145 zeros each (532 observed versus 578 expected in total).  At the preceding $|\Pi|=16,384$ scale the coordinate mean remained 21.25--21.68, the per-frame zero-minimum upper bound was at most $4.34\times10^{-6}$, and no zero occurred.  This extreme-value transition, rather than a failure of the coordinate expectation, explains the single scale pair whose minimum-impact ratio is undefined.

\paragraph{Device utilization.} Five-second telemetry recorded mean GPU utilization 100.0\%, mean allocated memory 41.46~GiB per device, and mean board power 296.2~W; the maximum observed temperature was 76$^\circ$C.  These measurements are retained as observations rather than folded into the replay summary.  Warp aggregation reduced the number of global packet/byte impact update pairs by only 0.124\% on average (maximum 0.622\%), confirming the collision-sparse regime of Corollary~\ref{cor:warp-sparse}.  It preserved both conservation identities but supplied no material production benefit; the separately measured 5.54\% crossover gain belongs to 32-bank sharding. The finite-stream deviation verifier found maximum resident-cardinality rounding error 0.80 blocks (0.000195 of the block set), mean absolute coordinate-level theory error 1.1725\%, and maximum normalized residual 5.95 standard-error units over 2,890,144 resident coordinates.  In the class containing that maximum, the RMS normalized residual was 1.000 and the 95th percentile was 1.977. The largest relative percentage error occurred in the $|\Pi|=65536$, $r=8$, $q=0.9$ cold class, where the expected impact was only 0.0024~ppb; the verifier therefore reports it as a finite-count tail effect rather than as a resident-mass error.
% END GENERATED GPU FINDINGS

\section{Confluent invalidation and a replicated index}
\label{sec:confluence}

Invalidation drops residency from a seed and forwards that change across transfer links.  To model asynchronous propagation as a replicated index, fix $S\subseteq \V$ and let
\[
F(X)=X\cup\{\,v\in \V : \exists u\in X,\ u\to v\in E\,\}.
\]

\begin{theorem}[Confluence and coordination-freeness]
\label{thm:confluence}
The invalidated states ordered by inclusion form a join-semilattice under $\cup$, on which $F$ is monotone and inflationary; its least fixed point above $S$ is the forward-reachable set $R=\{v:\text{$v$ is $E$-reachable from $S$}\}$. Every fair execution from $S$, under any order, batching by union, or duplication, eventually stabilizes at $R$. Invalidation is therefore confluent and coordination-free in the sense of monotone eventual consistency.
\end{theorem}

\begin{proof}
The join $\cup$ is commutative, associative and idempotent, so the states form a join-semilattice; $F$ is monotone and inflationary, so by the Knaster--Tarski theorem~\cite{tarski} it has least fixed point $\bigcup_{n\ge0}F^n(S)=R$ above $S$. Any execution adds only nodes of $R$ (a new node has an invalidated in-neighbour, inductively in $R$), and fairness adds every node of $R$ (each has a predecessor added first along a shortest path). Idempotence makes duplication and reordered batching immaterial. Monotonicity with the semilattice structure is exactly the condition for coordination-free convergence~\cite{calm,crdt,ameloot}; an invalidated location is a tombstone, moving only upward in the order.
\end{proof}

\begin{theorem}[Delta-state replicated index]
\label{thm:delta}
Represent each replica's invalidated set as a grow-only set and propagate \emph{deltas}, the newly invalidated identifiers, by the join $\cup$.  Under fair asynchronous delivery, every replica eventually converges to $R$, despite reordering or duplication.  If the replica overlay has diameter $d$ and a synchronous round delivers across every overlay edge, convergence occurs within $d$ rounds after the last new identifier is introduced.  The payload emitted by a replica in one round is bounded by the identifiers newly added since its previous emission.
\end{theorem}

\begin{proof}
A grow-only set under $\cup$ is a state-based CRDT whose merges are commutative, associative and idempotent~\cite{crdt,delta}. Fair delivery eventually carries every delta to every replica, so all states converge to their global join, which is $R$ by Theorem~\ref{thm:confluence}.  No finite bound follows for an arbitrary asynchronous schedule: fairness may delay a message for an unbounded time.  Under the stated synchronous full-edge schedule, an identifier advances at least one edge per round and reaches every replica along a shortest path of length at most $d$.  Sending only additions gives the payload bound.
\end{proof}

Theorem~\ref{thm:delta} supplies order-independence and bounded delta payloads for eventual invalidation.  Certificates computed before convergence are local snapshots; after every replica has joined all deltas they agree on the stable set $R$.

\section{Recoverable versus irrecoverable eviction}
\label{sec:reinstate}

Eviction is of two kinds. A \emph{recoverable} eviction demotes a block to a slower tier from which it can be reloaded; an \emph{irrecoverable} one releases the block for reuse and is absorbing.  The origin remains resident.  A configuration is therefore a block map
\[
\gamma:\Pi\longrightarrow\{\textsf{res},\textsf{rec},\textsf{free}\}.
\]
The operations $\mathsf{demote}(Q)$, $\mathsf{reload}(Q)$, and $\mathsf{free}(Q)$ act on a set of blocks $Q\subseteq\Pi$: demotion maps $\textsf{res}$ to $\textsf{rec}$, reload maps $\textsf{rec}$ to $\textsf{res}$, and free maps either non-free state to $\textsf{free}$.  Reload ignores freed blocks and free is absorbing.

\begin{theorem}[Reinstatement lattice]
\label{thm:reinstate}
Let $\gamma_0$ be a block configuration.
\begin{enumerate}[label=\textnormal{(\arabic*)},leftmargin=2.4em,itemsep=1pt]
\item \textnormal{(Recoverable regime.)} If $\gamma_0$ has no freed block and no schedule uses $\mathsf{free}$, the reachable configurations are exactly $\{\textsf{res},\textsf{rec}\}^{\Pi}$, a Boolean lattice under the blockwise order $\textsf{res}<\textsf{rec}$. Full residency is reached by $\mathsf{reload}(\Pi)$, and reachability of a block-aligned target configuration is decidable in $O(|\Pi|)$.
\item \textnormal{(Irrecoverable regime.)} The set of freed blocks is nondecreasing along every schedule and equals the initial freed set union the cumulative free domain. A block is reloadable iff it is currently recoverable, and entry into the freed set is permanent.
\end{enumerate}
\end{theorem}

\begin{proof}
(1) With no free operation, each block automaton has both transitions between $\textsf{res}$ and $\textsf{rec}$, independently of the other blocks. Hence the reachable set is the stated product lattice, and one reload over all blocks reaches full residency. (2) Free only adds blocks to the freed set and no transition removes them. Thus after a schedule with cumulative free domain $D$, the freed set is $\operatorname{Free}(\gamma_0)\cup D$, and every freed block remains fixed.
\end{proof}

The recoverable regime keeps full block residency within reach; an irreversible release removes that block from every future resident configuration.  Expressing the automaton over $\Pi$, rather than over individual locations, preserves alignment by construction.

\section{Discussion}
\label{sec:summary}

The static decision contract is exact.  For one candidate, two reachability forests and their closed boundaries certify the demand removed by that block.  For a two-hop network, the same quantity reduces to a streaming sufficient statistic and all candidates can be accumulated together.  These are different complexity statements and should not be merged: the accelerated result does not imply a one-pass all-candidate algorithm for arbitrary directed graphs.

The replay record also has a precise scope.  It is lossless for the generated packet stream and for metrics that factor through the stored monoid summary.  It does not make wall-clock timing deterministic; timing and device state are measured observations and are retained separately as result records and telemetry.  The two fingerprints detect accidental disagreement but are not used as a proof of equality.  Equality of the reported impact metrics is checked against the complete stored vectors and conservation identities.

\paragraph{Scope and limits.} The $k$-competitive statement is confined to the block-reference paging specialization.  Path admission, weighted retrieval, and fully dynamic directed reachability require additional algorithms; this paper claims re-traversal for a general fixed candidate and a one-pass kernel only for the two-hop case.  The hardness theorem concerns unweighted block actions with protected terminals.  The experimental networks are synthetic and deliberately controlled.  In particular, the rows labelled LRU, FIFO, and random are topology-independent total orders, not cache policies replayed against a deployment trace.  Their ratios quantify the value of topology information under the stated generator; they do not estimate a universal deployment penalty for LRU or FIFO.  Warp aggregation is likewise an exact transformation rather than an empirical speed contribution in this regime: the measured reduction is only 0.124\% on average, whereas the separate 32-bank crossover accounts for the 5.54\% pilot gain.  Fair asynchronous gossip has no finite convergence-time bound; the diameter guarantee applies only to synchronous full-edge rounds.

\paragraph{Open empirical questions.}
The controlled separation identifies three measurements needed from operational networks.  First, how strongly are recency and insertion history correlated with service-cut impact, and when does that correlation make a topology-oblivious order competitive in practice?  Second, what target concentration or namespace size makes warp grouping worthwhile, and can a runtime choose among lane grouping, bank sharding, and direct atomics from an online collision estimate?  Third, does the zero-impact threshold of Proposition~\ref{prop:zero-threshold} persist under bursty, dependent traffic, or does an extremal-index correction replace the independent-packet law?  Answering these questions requires trace-linked topology and residency snapshots; the present artifact supplies the exact counters and verifier against which such traces can be compared.

\section{Conclusion}
\label{sec:conclusion}

Aligned eviction is a graph decision, not merely a recency decision.  The paper supplies an independently checkable impact certificate, separates polynomial location cuts from hard block-action minimization, and gives a one-pass data structure for the two-hop case used in the accelerated study.  Counter-based replay and exact additive summaries decouple experimental scale from archive size without discarding the candidate vectors, samples, logs, or telemetry needed for audit.  The CRDT formulation handles invalidation order, while the state model keeps recoverable demotion distinct from irreversible release.

\appendix

\section{Accelerated streaming algorithm and audit artifact}
\label{app:artifact}

\subsection{Exact-stratified residency}

\medskip
\noindent\textbf{Algorithm A.1: constant-state residency membership.}
\begin{enumerate}[label=\arabic*:,leftmargin=3.2em,itemsep=1pt]
\item Split the $B$ blocks into hot and cold classes of size $H=B/2$, where $H$ is a power of two.
\item For each class derive an odd multiplier $a$ and offset $c$ from the scenario seed. Store these four words and $k=\operatorname{round}(qH)$.
\item For queried block $b$, compute its class and local index $x$, then rank $y=(ax+c)\mathbin{\&}(H-1)$.
\item Return resident iff $y<k$.  Because $a$ is odd, the rank map is a permutation and each class contains exactly $k$ resident blocks.
\end{enumerate}

\subsection{Streaming accumulator}

The implementation uses 64-bit unsigned counters.  The verifier checks
\[
N\le\left\lfloor\frac{2^{64}-1}{1500}\right\rfloor,
\qquad
64N\le B_{\rm total}\le1500N,
\]
so neither a packet counter nor a byte counter can overflow under the encoded packet-size range.  It also checks that no per-block impact exceeds the corresponding stream total. Parent identifiers are deduplicated before the resident-parent count is taken; without this step two equal route draws would be mistaken for path redundancy.

\medskip
\noindent\textbf{Algorithm A.2: one scenario on one device.}
\begin{enumerate}[label=\arabic*:,leftmargin=3.2em,itemsep=1pt]
\item Zero the 32 packet-impact banks, 32 byte-impact banks, service-state totals, and replay fingerprints.  At the maximum $|\Pi|=65{,}536$, the two bank arrays occupy exactly $2\cdot65{,}536\cdot32\cdot8=32$~MiB.
\item Until the five-minute deadline, launch a counter kernel over the 41~GiB buffer.  At position $j$, write $g_{\theta,s}(a+j)$ and advance the global counter interval after the pass.
\item In parallel for each generated packet, derive at most $r$ parent blocks, remove duplicate parents, and evaluate the deterministic residency predicate for each distinct parent.
\item If no parent is resident, increment the unserved counter.  If exactly one parent $b$ is resident, group warp lanes by $b$, map the warp to one of 32 banks, then let one lane atomically add the group size and group byte sum to that bank of $b$.  Otherwise increment the multiply-served counter.
\item Fold the packet count, byte count, and two 64-bit fingerprints into the scenario summary using atomic commutative operations.
\item Synchronize once per buffer pass.  At the deadline, reduce each block's 32 banks in a deterministic kernel, then copy the two impact vectors and scalar summary to the host.
\item Among resident blocks, select the minimum packet-impact victim.  Select LRU, FIFO, and random victims from topology-independent deterministic metadata.
\item Append one JSON result, one binary detail frame, and 2,048 evenly spaced decoded audit samples; flush both files before starting the next scenario.
\end{enumerate}

\subsection{Canonical binary detail frames}

Each device first writes one canonical little-endian detail stream.  Keeping this measurement path uncompressed isolates the timed kernel from the artifact codec.  The file begins with the eight-byte magic string \texttt{AEIMPACT}, a schema number, an endian marker, and the audit-sample count.  Every schema-3 scenario frame then stores a tag; scenario and repetition identifiers; block count, redundancy, and occupancy; seed and half-open counter interval; packet and byte totals; both fingerprints; the exactly-one-parent byte total; the number of warp-aggregated impact groups; the complete packet-impact vector; the complete byte-impact vector; and 2,048 audit triples.  An audit triple contains the counter, generated 64-bit event, and packed decoded metadata (packet bytes, live-parent count, and the sole parent identifier or a sentinel when that count is not one).

The frame size is exactly
\[
4+5(4)+9(8)+16|\Pi|+2048(24)
=49{,}248+16|\Pi|\ \text{bytes}.
\]
Across four replicates of the $3\times4\times3$ matrix, the detail streams therefore occupy
\[
144(49{,}248)+16\cdot48(4096+16384+65536)
=73{,}152{,}000\ \text{bytes},
\]
plus three 24-byte file headers.  This bound includes every candidate impact and all 294,912 decoded audit records before lossless coding.  The 50,000,000-byte ancillary budget is checked over the implementation, logs, calibration checkpoint, production vectors, telemetry, and validation output together.

\subsection{Reconstructible lossless artifact codec}

The production streams also have a compact representation, identified by the eight-byte magic string \texttt{AECODEC1}.  The transformation is performed after measurement and is exactly reversible to the schema-3 frame.  It exploits two invariants already checked by the verifier rather than discarding observations.

\paragraph{Residency-conditioned vector coding.}
For either impact vector $x\in\{0,\ldots,2^{64}-1\}^{|\Pi|}$, every nonresident coordinate is zero.  In each of the hot and cold classes, let $m$ be the lower median of the resident coordinates.  Resident coordinates are visited in increasing block order and represented by
\[
d_b=x_b-m,\qquad
\operatorname{zz}(d)=
\begin{cases}
2d,&d\ge0,\\
-2d-1,&d<0,
\end{cases}
\]
followed by unsigned LEB128 coding of $\operatorname{zz}(d_b)$.  The two medians are stored as 64-bit values.  A high bit in each payload-length field selects a fixed-width fallback: if the residual payload is not shorter than the original vector, the codec stores that vector verbatim.  Thus the choice is made independently for packet and byte impacts and cannot expand either vector payload.

\paragraph{Seed-indexed audit coding.}
For sample index $j\in\{0,\ldots,K-1\}$, $K=2048$, the canonical counter is
\[
c_j=a+\left\lfloor
\frac{(2j+1)(b-a)}{2K}
\right\rfloor
\]
for scenario interval $[a,b)$, and the event is $g_{\theta,s}(c_j)$.  The codec therefore stores only the eight-byte packed metadata for each sample.  The decoder reconstructs both the counter and event exactly from the frame fields and restores the original 24-byte audit triple.

\begin{theorem}[Byte-exact decodability]
\label{thm:artifact-codec}
For every valid schema-3 frame, the artifact codec reconstructs the complete packet-impact vector, byte-impact vector, and every audit triple byte for byte.  The encoding is injective on valid frames.  Its fixed per-frame cost is
\[
16{,}552+L_{\rm pkt}+L_{\rm byte}\ \text{bytes},
\]
where each $L$ is the smaller of its residual payload and the corresponding $8|\Pi|$-byte fixed-width vector.  Consequently a coded frame is at most
\[
16{,}552+16|\Pi|
\]
bytes, saving at least 32,696 bytes relative to its canonical schema-3 frame.  This guarantee is independent of outer ZIP compression.
\end{theorem}

\begin{proof}
The seed and exact-stratified predicate identify every resident coordinate; all remaining coordinates are zero by the one-pass update rule.  ULEB128 is prefix-decodable, ZigZag is a bijection from signed integers to nonnegative integers, and adding the stored class median inverts the residual transform.  The fixed-width mode is the identity.  The displayed sample formula and deterministic generator reconstruct the two omitted audit fields, while packed metadata is stored verbatim.  Hence the decoder recovers the canonical field sequence uniquely.  The coded fixed fields occupy 96 bytes, predictor and length fields 40 bytes, sample metadata $2048\cdot8=16{,}384$ bytes, and the frame digest 32 bytes, totaling 16,552 bytes before vector payloads.  Each adaptive payload is no larger than $8|\Pi|$, which proves the bound and the 32,696-byte saving against $49{,}248+16|\Pi|$.
\end{proof}

Each coded frame carries SHA-256 of the reconstructed canonical frame.  This digest is an error detector, not the reason the transform is lossless: injectivity follows from the explicit inverse above.  The independent decoder first reconstructs schema-3 bytes and checks the digest, then applies the semantic conservation, victim, residency, and replay checks.

\paragraph{Pre-production size pilot.}
On the 36-configuration one-second matrix, the codec reduced the canonical detail stream from 18,288,024 to 3,643,748 bytes, a factor of 5.02, while the independent analyzer emitted an identical summary CSV.  Deflate level 9 reduced the canonical and coded streams to 4,775,820 and 2,931,493 bytes, respectively; the codec therefore retained a further factor of 1.63 after outer compression.  These figures characterize the artifact representation and are not used as performance results.

\subsection{Independent validation}

\medskip
\noindent\textbf{Algorithm A.3: artifact verifier.}
\begin{enumerate}[label=\arabic*:,leftmargin=3.2em,itemsep=1pt]
\item Parse all JSON records and coded binary frames; reconstruct canonical schema-3 frames, check each SHA-256 digest, and reject duplicate or missing $(\text{scenario},\text{repetition})$ keys and malformed counter intervals.
\item Match every scalar field and fingerprint in JSON against its binary frame.
\item Verify the conservation identity
\[
\sum_{b\in\Pi} I_{\rm packet}(b)
=N_{\rm one},
\qquad
\sum_{b\in\Pi} I_{\rm bytes}(b)
=B_{\rm one},
\]
because every packet with exactly one resident parent contributes its packet and byte weight to exactly one block and every other packet contributes to none.
\item Recompute the resident predicate for every block, take the minimum stored impact over resident candidates, and compare it with the reported impact-priced victim.
\item Look up each reported LRU, FIFO, and random victim in both stored impact vectors and compare its packet and byte values with JSON.
\item For every audit sample, regenerate the event from its seed and counter, reconstruct distinct parents, residency, packet size, live-parent count, and sole parent, then compare the packed metadata bit for bit.
\item Aggregate repetitions only after all checks pass.  Emit the CSV used by Table~\ref{tab:gpu} and a machine-readable validation report.
\end{enumerate}

\subsection{Deviation-envelope extraction}

The verifier also emits \texttt{deviation\_summary.csv} and \texttt{deviation\_report.json}.  These files are derived from the same reconstructed impact vectors used for victim validation; no additional measurement path is introduced.  The purpose is to separate four effects that are otherwise easy to confuse: resident-cardinality rounding in Theorem~\ref{thm:exact-residency}, finite-stream route variation around Proposition~\ref{prop:expected-impact}, extreme-value collapse of the minimum in Proposition~\ref{prop:zero-threshold}, and genuine policy damage relative to the impact-priced victim.

\medskip
\noindent\textbf{Algorithm A.4: deviation-envelope verifier.}
\begin{enumerate}[label=\arabic*:,leftmargin=3.2em,itemsep=1pt]
\item For each validated frame, recompute the exact resident set and resident probability mass $R$ from the stored seed, block count, and occupancy.
\item For every resident block $b$, compute $\zeta_b=(1-R+\pi_b)^r-(1-R)^r$ and residual $e_b=X_b-N\zeta_b$ from the stored packet-impact vector.
\item Aggregate residuals separately for the hot and cold classes.  For each class write the mean observed ppb, theory ppb, signed relative error, mean and maximum absolute relative error, and the envelope $\Delta_C=\max |e_b|$ both in packets and ppb.
\item Also write normalized coordinates $e_b/\sqrt{N\zeta_b(1-\zeta_b)}$ when $\zeta_b>0$.  These values are diagnostics for finite-count scale; the artifact treats the exact residual envelope, not a distributional assumption, as the certificate.
\item Count zero-impact coordinates in each class and compute $\mu_C=m_C(1-\zeta_C)^N$, the probability bracket in Proposition~\ref{prop:zero-threshold}, and the threshold margin $N\zeta_C-\log m_C$.
\item Record the realized resident fraction and its deviation from the target $q$.  Because each class rounds independently, the total resident-count error is at most one block and the fraction error is at most $1/|\Pi|$.
\item Emit a compact JSON summary containing the worst relative coordinate error, worst normalized residual, total resident-coordinate count, zero-minimum transition, maximum resident-cardinality rounding error, and maximum class envelope.  The generated findings paragraph reads this JSON after the full 144-scenario matrix validates.
\end{enumerate}

\subsection{Ancillary evidence}

The ancillary files contain the implementation, reconstructible exact impact vectors, audit-sample metadata, telemetry, build and progress logs, manifests, validation output, and deviation summaries.  Production detail streams use the byte-exact \texttt{AECODEC1} representation; the Bernoulli calibration checkpoint remains in its original schema.  The packet stream itself is represented by generator parameters and counter intervals as specified by Theorem~\ref{thm:replay}.  Validation reconstructs every canonical frame, checks its digest and semantic invariants, and verifies the recorded file hashes.

\end{document}